\documentclass[a4paper,UKenglish,cleveref, autoref, thm-restate]{lipics-v2021}

\usepackage[linesnumbered,ruled,vlined,resetcount,noend]{algorithm2e}

\nolinenumbers
\title{Locating Charging Stations: Connected, Capacitated and Prize- Collecting} 

\titlerunning{Locating Charging Station Problem} 

\author{Rajni Dabas}{Department of Computer Science, University of Delhi, [Delhi-110007], India }{rajni@cs.du.ac.in}{}{Supported by a UGC-JRF}

\author{Neelima Gupta}{Department of Computer Science, University of Delhi, [Delhi-110007], India }{rajni@cs.du.ac.in}{}{}

\authorrunning{R. Dabas and N. Gupta} 

\Copyright{Rajni Dabas and Neelima Gupta} 


\keywords{Facility Location, Connected Facility Location, Capacitated Facility Location, Prize Collecting Facility Location, Penalties, Lower Bounds} 

\category{} 

\relatedversion{} 

\EventEditors{}
\EventNoEds{2}
\EventLongTitle{}
\EventShortTitle{}
\EventAcronym{}
\EventYear{}
\EventDate{}
\EventLocation{}
\EventLogo{}
\SeriesVolume{}
\ArticleNo{}

\newcommand{\facilityset}{\mathcal{F}}	
	
\newcommand{\clientset}{\mathcal{C}}			
\newcommand{\cliset}{\mathcal{C}}				

\newcommand{\dist}[2]{c(#1,~#2)}
\newcommand{\concost}[1]{d_{#1}}

\newcommand{\etal}{\textit{et al}.}

\newcommand{\clipen}{\mathcal{C}_p}

\newcommand{\dummyset}{\mathcal{D}}

\newcommand{\Opt}[1]{\mathcal{O}_{#1}}

\newcommand{\edgeset}{\mathcal{F} \times \mathcal{F}}
\newcommand{\CPFL}{CPFL }
\newcommand{\cpfl}{cpfl}
\newcommand{\CFL}{CFL }
\newcommand{\cfl}{cfl}

\newcommand{\ConFL}{ConFL }
\newcommand{\confl}{confl}
\newcommand{\ConCPFL}{ConCPFL }
\newcommand{\concpfl}{concpfl}
\newcommand{\ConPFL}{ConPFL }
\newcommand{\conpfl}{conpfl}
\newcommand{\ConCFL}{ConCFL }
\newcommand{\concfl}{concfl}
\newcommand{\ConkM}{Con$k$M }
\newcommand{\ConkFL}{Con$k$FL }
\newcommand{\ConCkFL}{ConC$k$FL }

\newcommand{\CPkFL}{CP$k$FL }
\newcommand{\ConCkM}{ConC$k$M }

\begin{document}
\maketitle

\begin{abstract}
In this paper, we study locating charging station problem as facility location problem and its variants ($k$-Median, $k$-Facility location and $k$-center). We study the connectivity and the capacity constraints in these problem.

Capacity and connectivity constraints have been studied in the literature separately for all these problems. We give first constant factor approximations when both the constraints are present. Extending/modifying the techniques used for connected variants of the problem to include capacities or for capacitated variants of problem to include connectivity is a tedious and challenging task. In this paper, we combine the two constraints by reducing the problem to underlying well studied problems, solving them as black box and combine the obtained solutions. We also, combine the two constraints in the prize-collection set up. 

In the prize-collecting set up, the problems are not even studied when one of the constraint is present. We present constant factor approximation for them as well.
\end{abstract}

\section{Introduction}
To address the increasing environmental health issues, countries around the globe are planning to phase out combustion engine vehicles. Automobile manufacturers, like Nissan Leaf, Tesla, Mahindra and Tata Motors, are switching to produce electric vehicles. One of the major challenge posed by this shift is to strategically identify the locations to set up the charging stations. Due to the short range of electric vehicles, the existing refueling station model is not sufficient. Governments across the world are trying to address the issue, but, high costs associated with equipment and installation of charging stations at public spaces are currently obstructing the build-out of such a network. So, what we need is a cost effective solution for locating the charging stations. 

The problem of locating charging stations can be formulated as a facility location problem. {\em Facility Location Problem} (FL) is a well known and well studied problem in operations research and computer science~\cite{Shmoys,ChudakS03,Jain:2001,Byrka07,guha1999greedy,KPR,Chudak98,mahdian2001greedy,mahdian_1.52,charikar2005improved,arya,Zhang2007126,Li13}. 
In (uncapacitated) FL, we are given a set of facilities and a set  of clients. Every facility has an associated {\em facility opening cost}. Serving a client from a facility incurs {\em service cost}. We assume that the service cost is a metric.
The goal is to open a subset of facilities so as to minimise the sum of facility opening costs ({\em facility cost}) and the total {\em service cost} of serving all the clients from the opened facilities. In case of charging stations, the facilities are charging stations and the consumers are clients. Another closely related problem is $k$-Median problem ($k$M). In (uncapacitated) $k$M, we are given an upper bound $k$ ({\em cardinality constraint}) on the maximum number of facilities that can be opened in our solution and facility opening costs are $0$ for every facility. $k$-Facility Location is a common generalisation of FL and $k$M where in we have both the facility opening costs and the cardinality constraint. Yet another variant is, $k$-center problem ($k$C) where we wish to minimise the maximum distance of any client from the serving facility.

The problems are NP-hard. Several approximation results are known for these basic problems. For example an elegant $2$ factor approximation was obtained for $k$-Center more than three decades ago by Hochbaum and Shmoys ~\cite{hochbaum1985kC} 
and the factor is known to be tight. Since then most of the work has focused on FL and $k$M problems.  For facility location the gap between the best known approximation ratio (1.488)~\cite{Li} 
and the best known lower bound (1.463)~\cite{guha1999greedy}
is nearly closed, whereas it is yet to be filled for the $k$ median problem with the current best approximation ratio and the best known lower bound being $2.675 + \epsilon$~\cite{ByrkaPRST14} 
and $\approx 1.736$~\cite{jain2002new} 
respectively.

The problems become harder as constraints are added to them.  Constraints come naturally in locating the charging stations. 
One such constraint is the limited number of charging slots at a given station. In the context of FL, the problem with a bound on the maximum number of clients that a facility can serve, is called the capacitated facility location (CFL).  This constraint is notoriously hard to handle. The same has been accepted in the literature by many researchers. For example, An~\etal~\cite{AnCkC} states "there is a large discrepancy in our understanding of uncapacitated and capacitated versions of network location problems". Though the discrepancy is bridged to some extent for problems like CFL and C$k$C (constant factor approximations for the problems have been achieved using local search~\cite{Aggarwal, Bansal} and LP-rounding with preprocessing~\cite{AnCkC} respectively), the results and the techniques that are successful in dealing with capacities are still very limited and no true constant factor approximation is known for C$k$M in the literature. 


Another constraint in the charging station problem arises due to the need to connect the charging stations 
for them to receive electric supply from the power grid. This puts a connectivity constraint on the opened facilities (the charging stations). Also, the connectivity must be acquired at low cost. Thus, this also adds a component of cost to the objective function.

Sometimes a few distant consumers (clients) can increase the cost of solution disproportionately. It is profitable for the company (installing the charging stations) to leave these clients unserved by paying some penalty cost. Another way to think about it is that every client has a prize that can only be collected if it is served. This generalisation of the FL problem is called prize-collecting FL (PFL).

R{\"o}sner and Schimdt~\cite{rosnerICALP} realise the difficulty to adjust individual methods for each added constraint and pose a challenging request to build a solution for the constrained instance using an existing solution for the  lesser constrained one, as a black box, satisfying at least one more constraint.
In other words, they reduce the constrained instance of a problem to a lesser constrained one and use the solution of the latter as a  black-box, to obtain the solution for the original problem. 
They provide a framework to add privacy in variants of clustering with different constraints. In this paper, we provide a framework to add connectivity to variants of charging station problem with different constraints. We also provide a framework to add penalties to some variants of the problem. In particular we study the following problems:

\begin{itemize}
    \item Connected Capacitated Facility Location (ConCFL)
    \item Connected Capacitated $k$-Median (ConC$k$M)
    \item Connected Capacitated $k$-Facility Location (ConC$k$FL)
    \item Connected Capacitated $k$-Center (ConC$k$C)
    \item Connected Prize Collecting Facility Location (ConPFL)
    \item Capacitated Prize Collecting Facility Location (CPFL)
    \item Connected Capacitated Prize Collecting Facility Location (ConCPFL)
     \item Connected Capacitated Prize Collecting $k$-Median (ConCP$k$M)
      \item Connected Capacitated Prize Collecting $k$-Facility Location (ConCP$k$FL)
\end{itemize}

\begin{figure}
    \centering
    \includegraphics[width=12cm]{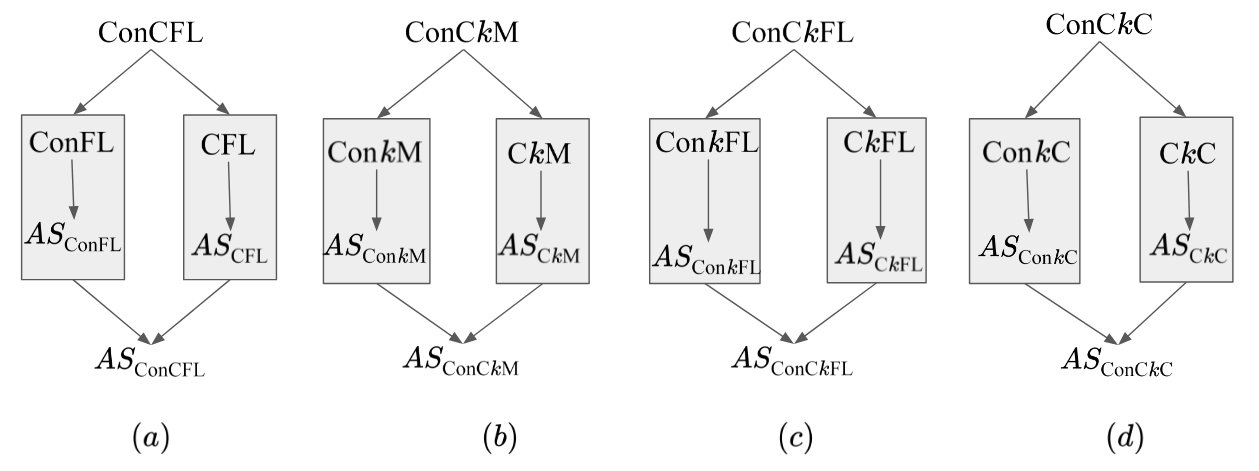}
    \caption{Combining connectivity and capacities. Using ConFL instead of Con$k$M and Con$k$FL as black-box in (b) and (c) will also work.}
   \label{figConC}
\end{figure}

\begin{figure}
    \centering
    \includegraphics[width=12cm]{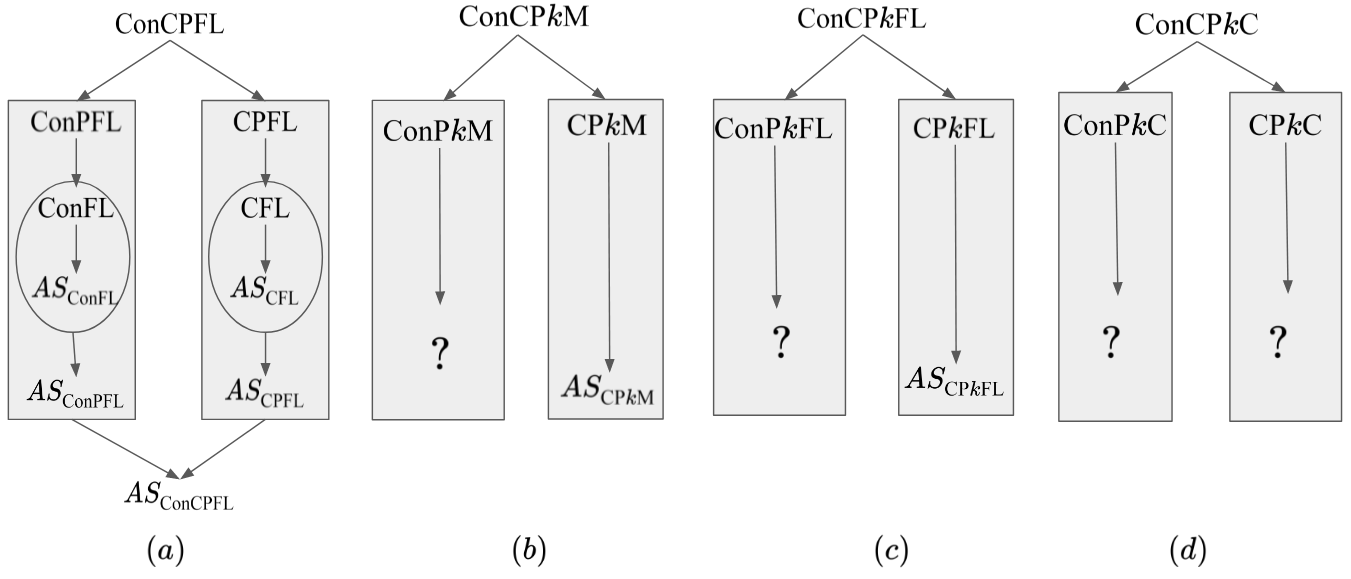}
  \caption{Combining connectivity, capacities and penalties.}
    \label{figConCP}
\end{figure}

\begin{figure}
    \centering
   \includegraphics[width=7cm]{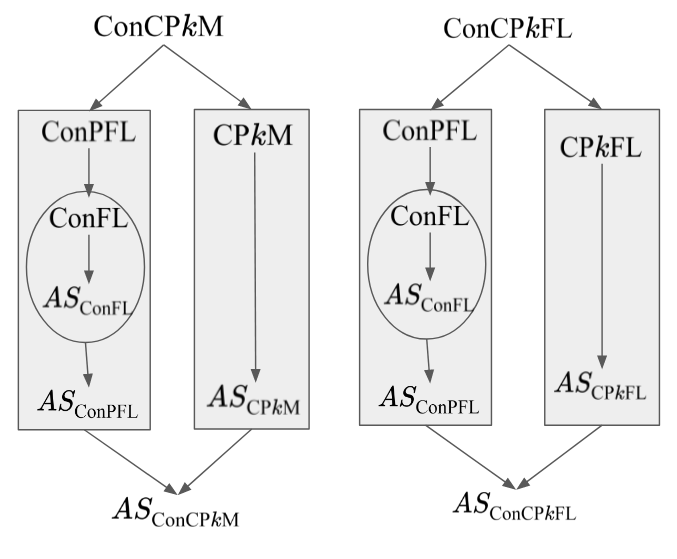}
    \caption{Combining connectivity, capacities and penalties. Using ConPFL instead of ConP$k$M and ConP$k$FL as black-box.}
    \label{figConCP2}
\end{figure}

\noindent 
The constraints, capacities and connectivity, have been studied separately for all the basic problems (FL, $k$M, $k$FL and $k$C). Extending/modifying the techniques used for connected variants of the problem to include capacities or for capacitated variants of problem to include connectivity is a tedious and challenging task. In this paper, we combine the two constraints using reduction to underlying well studied problems. Instead of extending and modifying the solutions to the underlying problems, we use them as black box and combine the obtained solutions. We also, combine the two constraints in the prize-collection set up.
Figures~\ref{figConC} and~\ref{figConCP} give a broad idea of our plan. Figure $1$ depicts the reduction of Connected - Capacitated variant of the problems to two problems, each with only one constraint, either Connectivity or Capacities.  Figure $2$ depicts the reduction of Connected - Capacitated - Prize-Collecting variant of the problems to two problems:
Connected - Prize Collecting and Capacitated - Prize Collecting variants of the problem. Figure $2$ also depicts the reduction of Connected/Capacitated - Prize Collecting variant to Connected/Capacitated variant of FL. The last reduction does not work in presence of the cardinality constraint as depicted in Figure~\ref{figConCP} (b) and (c) i.e., we are not able to reduce ConP$k$FL/ConP$k$M to Con$k$FL/Con$k$M. However, we are able to solve ConCP$k$FL/ConCP$k$M by reducing them to ConPFL instead, as shown in Figure~\ref{figConCP2}. In particular, we present the following results:

\begin{theorem}
\label{thm_ConCFL}
Given an $\alpha$-factor approximation for for ConFL and a $\beta$-factor approximation for CFL/C$k$M/C$k$FL with $\gamma$-factor violation in capacity/cardinality, a $(\alpha+2\beta)$-factor approximation for ConCFL/ConC$k$M/ConC$k$FL preserving the violations in capacities/cardinality can be obtained in polynomial time.
\end{theorem}

\begin{theorem}
\label{thm_ConCkC}
Given an $\alpha$-factor approximation for C$k$C and a $\beta$-factor approximation for Con$k$C, 
a $(2\alpha+\beta)$-factor approximation for ConC$k$C can be obtained in polynomial time.
\end{theorem}

\begin{theorem}
\label{thm_ConCPFL}
Given an $\alpha$-factor approximation for ConPFL and a $\beta$-factor  approximation for CPFL/CP$k$M/\CPkFL with $\gamma$-factor violation in capacities, a $(\alpha+2\beta)$-approximation for ConCPFL/ConCP$k$M/ConCPkFL preserving the violations in capacities/cardinality can be obtained in polynomial time.
\end{theorem}

\begin{theorem}
\label{thm_ConPFL}
Given an $\alpha$-factor approximation for ConFL (using LP optimal to lower bound the cost of the optimal), a $2\alpha$-factor approximation for ConPFL can be obtained in polynomial time.
\end{theorem}

\begin{theorem}
\label{thm_CPFL}
Given an $\beta$-factor approximation for CFL, an $\beta$-factor approximation for CPFL can be obtained in polynomial time.
\end{theorem}

To the best of our knowledge, our results are the first approximations for ConC$k$FL, ConC$k$M, ConC$k$C, ConCPFL, ConCP$k$M, ConCP$k$FL and ConPFL. The only known result for ConCFL in the literature is by Friggstad~\etal~\cite{FriggstadConLUFL2016}. They gave a constant factor approximation for lower and (uniform) upper bounded ConFL violating both the upper and the lower bounds. We cannot get rid of the  violations in the upper bounds even when there are no lower bounds using their technique as they use LP rounding and the LP is known to have an unbounded integrality gap. We give the first true constant factor approximation for the problem. The only result known for CPFL is by Gupta and Gupta~\cite{GuptaG14} using local search technique. They give $(5.83 + \epsilon)$ factor for the case of uniform capacities
and $(8.532 + \epsilon)$ factor for non-uniform capacities. Our result is interesting as it is much simpler and uses the underlying problem CFL as a black box without increasing the cost. We also improve the factor to $5$. The factor will improve, if the factor for the underlying problem (CFL) improves in future. 

Tables~\ref{tab_results1} and \ref{tab_results2} summarise the results we obtain by plugging in the best results known for the underlying problems. Approximation guarantees will improve if better solutions are obtained for the underlying problems.

\begin{table}[ht]
\small
    \centering
    \begin{tabular}{|c  c| c | c  c |}
        \hline
         
         \multicolumn{2}{|c|}{\textbf{Problem}}  &  \textbf{Sub-Problem 1} & \multicolumn{2}{c|}{\textbf{Sub-Problem 2}}      \\ 
         
         
        Factor & Violation& Factor & Factor & Violation\\
         
         \hline
         
         \multicolumn{2}{|c|}{\textbf{(U)ConCFL}} & \textbf{ConFL~\cite{ConFLGrandoni2011}} & \multicolumn{2}{c|}{\textbf{(U)CFL~\cite{Aggarwal}}} \\
          $9.19$ & Nil & $3.19$ & $3$ & Nil \\
        \hline
         
         \multicolumn{2}{|c|}{\textbf{(NU)ConCFL}}  & \textbf{ConFL~\cite{ConFLGrandoni2011}} & \multicolumn{2}{c|}{\textbf{(NU)CFL~\cite{Bansal}}} \\
         $13.19$ & Nil & $3.19$ & $5$ & Nil   \\
        \hline
         
         \multicolumn{2}{|c|}{\textbf{(U)ConC$k$M}} & \textbf{ConFL~\cite{ConFLGrandoni2011}} & \multicolumn{2}{c|}{\textbf{(U)C$k$M~\cite{ByrkaRybicki2015}}} \\ 
        $O(1/\epsilon)$  & $(1+\epsilon)u$  & $3.19$ & $ O(1/\epsilon)$ & $(1+\epsilon)u$  \\
        \hline 
         
         \multicolumn{2}{|c|}{\textbf{(NU)ConC$k$M}} & \textbf{ConFL~\cite{ConFLGrandoni2011}} & \multicolumn{2}{c|}{\textbf{(NU)C$k$M~\cite{Demirci2016}}}\\
         $O(1/\epsilon)$  & $(1+\epsilon)u_i$ & $3.19$ & $O(1/\epsilon)$ & $(1+\epsilon)u$ \\
         \hline 
         
         \multicolumn{2}{|c|}{\textbf{(U)ConC$k$M}} & \textbf{ConFL~\cite{ConFLGrandoni2011}} & \multicolumn{2}{c|}{\textbf{(NU)C$k$M~\cite{capkmshili2014}}}\\
         $O(1/\epsilon)$  & $(1+\epsilon)k$ & $3.19$ & $O(1/\epsilon)$ & $(1+\epsilon)k$ \\
         \hline
         
         \multicolumn{2}{|c|}{\textbf{(NU)ConC$k$M}} & \textbf{ConFL~\cite{ConFLGrandoni2011}} & \multicolumn{2}{c|}{\textbf{(NU)C$k$M~\cite{Lisoda2016}}}\\
         $O(1/\epsilon)$  & $(1+\epsilon)k$ & $3.19$ & $O(1/\epsilon)$ & $(1+\epsilon)k$ \\
         \hline
         
         \multicolumn{2}{|c|}{\textbf{(U)ConC$k$FL}} & \textbf{ConFL~\cite{ConFLGrandoni2011}} & \multicolumn{2}{c|}{\textbf{(U)C$k$FL~\cite{capkmByrkaFRS2013}}}\\
         $O(1/\epsilon^2)$ & $(2+\epsilon)u$ & $3.19$ & $O(1/\epsilon^2)$ & $(2+\epsilon)u$ \\
         \hline
         
        \multicolumn{2}{|c|}{ \textbf{(U)ConC$k$C}}  & \textbf{Con$k$C~\cite{ge2008ConkC}} & \multicolumn{2}{c|}{\textbf{(U)C$k$C~\cite{KhullerSussmanCkC}}}\\
        $18$ & Nil & $6$ & $6$ & Nil   \\
        \hline
        
        \multicolumn{2}{|c|}{\textbf{(NU)ConC$k$C}} & \textbf{Con$k$C~\cite{ge2008ConkC}} & \multicolumn{2}{c|}{\textbf{(NU)C$k$C~\cite{AnCkC}}}\\
       $24$ & Nil & $6$ &  $9$ & Nil  \\
         \hline
    \end{tabular}
\caption{Results corresponding to Theorems~\ref{thm_ConCFL} and~\ref{thm_ConCkC} based on current best approximation for the underlying problem. (U-Uniform, NU-Non-Uniform)}
    \label{tab_results1}
\end{table}

\begin{table}[ht]
\small
    \centering
    \begin{tabular}{|c  c| c | c  c |}
        \hline
         
         \multicolumn{2}{|c|}{\textbf{\footnotesize{Problem}}}  &  \textbf{\small{Sub-Problem 1}} & \multicolumn{2}{c|}{\textbf{\small{Sub-Problem 2}}}      \\ 
         
         
        Factor & Violation& Factor & Factor & Violation\\
         \hline

         \multicolumn{2}{|c|}{\textbf{(U/NU)ConCPFL}} & \textbf{ConPFL~[This Paper]} & \multicolumn{2}{c|}{\textbf{(U/NU)CPFL~[This Paper]}} \\
         $31.32$ & Nil & $21.32$ & $5$ & Nil \\
         \hline

         
         
         \multicolumn{2}{|c|}{\textbf{(U)ConCP$k$M}} & \textbf{ConPFL~[This Paper]} & \multicolumn{2}{c|}{\textbf{(U)CP$k$M~\cite{DabasCFLPO}}} \\
         
        $O(1/\epsilon^2)$ & $(2+\epsilon)u$ & $21.32$ &  $O(1/\epsilon^2)$ & $(2+\epsilon)u$  \\
         \hline

         \multicolumn{2}{|c|}{\textbf{(U)ConCP$k$FL}} & \textbf{ConPFL~[This Paper]} & \multicolumn{2}{c|}{\textbf{(U)CP$k$FL~\cite{DabasCFLPO}}} \\
         
        $O(1/\epsilon^2)$ & $(2+\epsilon)u$ & $21.32$  & $O(1/\epsilon^2)$ & $(2+\epsilon)u$ \\
         \hline
         
         \multicolumn{2}{|c|}{\textbf{ConPFL}} &  \textbf{ConFL~\cite{ConFLGupta2001}} & \multicolumn{2}{c|}{-} \\
         
         $21.32$ & Nil & $10.66$ & \multicolumn{2}{c|}{-} \\
        \hline
        
        \multicolumn{2}{|c|}{\textbf{(U/NU)CPFL}} & - & \multicolumn{2}{c|}{\textbf{(NU)CFL~\cite{Bansal}}} \\
        
        $5$& Nil & -& $5$ & Nil  \\
         \hline
    
    \end{tabular}
\caption{Results corresponding to Theorems~\ref{thm_ConCPFL},~\ref{thm_ConPFL} and~\ref{thm_CPFL} based on current best approximation for the underlying problem. (U-Uniform, NU-Non-Uniform) }
    \label{tab_results2}
\end{table}

\subsection{Our Techniques}
\noindent \textbf{The Combining Technique:} In ConCX/ConCPX problems (where X is FL, $k$M, $k$FL or $k$C as applicable), we reduce our problem to two sub-problems in a way that the connectivity constraint move to one sub problem (ConFL/ConPFL) and the capacity constraints to another (CX/CPX). We use the openings and the assignments of the solution of CPX and connect them using the connectivity of the solution of ConFL via clients.

\noindent \textbf{ConPFL:} We reduce an instance of ConPFL to an instance of ConFL using LP rounding and thresholding techniques. A client paying penalty to an extent of at least half in LP optimal is removed by paying full penalty. Assignment of the remaining clients is raised so that they are served to full extent. The openings are raised accordingly. Note that the reduction does not work in presence of cardinality constraints as increasing the opening of facilities can violate the cardinality constraint by a factor of $2$.

\noindent \textbf{CPFL:} We reduce an instance of CPFL to an instance of CFL by creating a dummy facility, with capacity $1$, collocated with every client.  The facility opening cost of a dummy facility is equal to the penalty cost of the respective collocated client.

\subsection{Previous and Related Work}
\textbf{Capaciated variants of the problem (CFL, C$k$M, C$k$FL and C$k$C)}: Shmoys \etal~\cite{Shmoys} gave the first constant factor($7$) approximation when the capacities are uniform, with a capacity blow-up of $7/2$, by rounding the solution to standard LP. Grover \etal~\cite{GroverGKP18} 
reduced the capacity violation to $(1 + \epsilon)$ thereby showing that the capacity violation can be reduced to arbitrarily close to $0$ by rounding a solution to standard LP. An \etal~\cite{Anfocs2014} gave the first true constant factor(288) approximation for the problem (non-uniform), by strengthening the standard LP. Local search has been successful in dealing with capacities, several results~\cite{KPR, ChudakW99, mathp, paltree, mahdian_universal, zhangchenye} have been successfully obtained using local search with the current best being $5$-factor for non-uniform~\cite{Bansal} and $3$-factor for uniform capacities~\cite{Aggarwal}. No true approximation is known for C$k$M, till date. Constant factor approximations~\cite{Charikar,Charikar:1999,capkmshanfeili2014,ChuzhoyR05,capkmByrkaFRS2013,aardal2013,GroverGKP18} are known, that violate capacities or cardinality by a factor of $2$ or more. Strengthened LPs and dependent rounding techniques have been used to give constant factor approximation with small violations in capacities/cardinality. For uniform capacities Byrka \etal ~\cite{ByrkaRybicki2015} and for non-uniform capacities Demirci \etal~\cite{Demirci2016} gave the approximations with $(1+\epsilon)$ violation in capacities whereas Li~\cite{capkmshili2014, Lisoda2016} gave the approximations using at most  $(1 + \epsilon) k$ facilities for uniform as well as non-uniform capacities. For uniform C$k$FL, a constant factor approximation, with $(2 + \epsilon)$ violation in capacities, was given by Byrka \etal~\cite{capkmByrkaFRS2013} using dependent rounding. This was followed by two constant factor approximations by Grover \etal~\cite{GroverGKP18}, one with $(1 + \epsilon)$ violation in capacities and $2$-factor violation in cardinality and the other with $(2+\epsilon)$ factor violation in capacities only. For (uniform) C$k$C, a $10$-factor approximation was given by Bar-Ilan~\etal
~\cite{bar-IlanCkC} which was improved to $6$ by Khuller and Sussman~\cite{KhullerSussmanCkC}. For non-uniform capacities, Cygan~\etal~\cite{CyganCkC} gave a large constant factor approximation, improved to $9$ by An~\etal~\cite{AnCkC} which is also also the current best.


\noindent \textbf{Connected Variants of the problem (ConFL, Con$k$M, Con$k$FL and Con$k$C):} ConFL was first introduced by Gupta~\etal~\cite{ConFLGupta2001} where they gave a $10.66$-factor approximation for the problem using LP-rounding technique. Gupta~\etal~\cite{ConFLGupta2004} described
a random facility sampling algorithm for the problem giving $9.01$-factor approximation. Primal and dual technique was first used by Swamy and Kumar~\cite{ConFLSwamy2004} to improve the factor to $8.55$ which was then improved to $6.55$ by Jung~\etal~\cite{ConFLChwa2009}. Eisenbrand~\etal~\cite{ConFLEisenbrand2010} used random sampling to improve the factor to $4$. The factor was improved to $3.19$ using similar techniques by Grandoni and Rothvob~\cite{ConFLGrandoni2011} which is also the current best for the problem. Eisenbrand~\etal~\cite{ConFLEisenbrand2010} extends their algorithm to \ConkFL giving $6.98$-factor approximation for \ConkM and \ConkFL. For Con$k$C, Ge~\etal~\cite{ge2008ConkC} gave a $3$-factor approximation when $k$ is fixed and a $6$-factor approximation for arbitrary $k$. Liang~\etal~\cite{liang2016ConkC} also gave a simpler $6$-factor approximation for the problem.

\noindent \textbf{Prize-Collecting FL (PFL)}: For PFL, a $3$-factor approximation using primal dual techniques was given by Charikar~\etal~\cite{charikar2001algorithms} which was subsequently improved to $2$ by Jain~\etal~\cite{jain2003greedy} using dual-fitting and greedy approach. Wang~\etal~\cite{Wang2015penalties} also gave a $2$-factor approximation using a combination of primal-dual and greedy technique. Later Xu and Xu~\cite{xu2005lp} gave a $2+2/e$ using LP rounding. The factor was improved  to $1.8526$ by the same authors in~\cite{xu2009improved} using a combination of primal-dual schema and local search. For linear penalties Li~\etal~\cite{LiDXX15} gave a $1.5148$-factor using LP-rounding.

\noindent \textbf{Connected and Capacitated FL (ConCFL): } A constant factor approximation for (uniform) \ConCFL with violation in capacities follows as a special case from approximation algorithm of Friggstad~\etal~\cite{FriggstadConLUFL2016} on Connected Lower and Upper Bounded Facility Location problem.

\noindent \textbf{Capacitated and Prize-Collecting variants of the problem (CPFL and CP$k$FL)}: For \CPFL, Dabas and Gupta~\cite{DabasCFLPO} gave an $O(1/\epsilon)$ approximation with $(1+\epsilon)$ factor violation in capacities using LP rounding techniques and reduction to CFL. The only true approximation is due to~Gupta and Gupta~\cite{GuptaG14}. They uses local search to obtain a $5.83$-factor and $8.532$-factor approximation for uniform and non-uniform variants respectively. For \CPkFL, Dabas and Gupta~\cite{DabasCFLPO} gave an $O(1/\epsilon^2)$-approximation algorithm with $(2+\epsilon)$-factor violation in capacities.

\noindent \textbf{ConC$k$FL, ConC$k$M, ConC$k$C, ConCPFL, ConCP$k$M, ConCP$k$FL and ConPFL}: To the best of our knowledge, no result is known for these problems in the  literature.

\subsection{Organisation of the Paper}
In Section~\ref{sectionConCFL}, we present our combining technique to obtain constant factor approximations for \ConCFL, \ConCkM, \ConCkFL and ConC$k$C. In Section \ref{sectionConCPFL}, we present our results for \ConCPFL, ConCP$k$M and ConCP$k$FL. In Sections~\ref{sectionConPFL} and ~\ref{sectionCPFL}, we give the results for \ConPFL and \CPFL respectively. Finally, we conclude with future work in Section~\ref{conclusion}.
\section{Connected and Capacitated variants of the problem}
\label{sectionConCFL}
Let $V$ be a set of locations and $G$ be a complete graph on $V$. Let $c:V \times V \rightarrow\mathbb{R}^+$ be a metric cost function. Let $\facilityset \subseteq V$ and $\cliset \subseteq V$ be the set of facilities and clients respectively. Each facility $i \in \facilityset$ has an associated opening cost $f_i$ and serving a client $j$ from facility $i$ incurs a cost $\dist{i}{j}$. In facility location problem (FL), we wish to open $\facilityset' \subseteq \facilityset$ of facilities and compute an assignment function $\phi : \clientset  \rightarrow \facilityset'$. Our goal is to minimize the total cost of opening the facilities in $\facilityset'$ and serving the clients in $\clientset$.
\begin{itemize}
    \item \textbf{ConFL:} In this variant, we wish that the set $\facilityset'$ of the opened facilities is connected by a steiner tree.
    We call the total cost of the steiner tree, that is the sum of the costs on edges of the steiner tree,   as connection cost.
%
    Our goal now is to 
    minimise the total cost of opening the facilities, serving the clients and the connection cost.
    \item \textbf{CFL:} In this variant, a facility $i$ has a bound $u_i$ on the maximum number of clients it can serve  i.e., $|\phi^{-1}(i)| \leq  u_i$ for all $i \in \facilityset'$.
    
\end{itemize}

\noindent In this section, we present a constant factor approximation for \ConCFL which is a common generalisation of ConFL and CFL, that is, we have both connectivity and capacity constraints in facility location problem. The result is obtained by combining the results of ConFL and CFL. Let $I_{\concfl}$ be an instance of ConCFL. We first make two instances $I_{\confl}$ and $I_{\cfl}$ ConFL and CFL by dropping the capacity and the connectivity constraints respectively. Note that, an optimal solution $\Opt{\concfl}$ of \ConCFL forms a feasible solution for instances $I_{\confl}$ and $I_{\cfl}$. Hence the cost of the optimal solutions $\Opt{\confl}$ and $\Opt{\cfl}$ of $I_{\confl}$ and $I_{\cfl}$ respectively are bounded. Next we solve these instance using approximation algorithms for \ConFL and CFL to obtain two approximate solutions $AS_{\confl}$ and $AS_{\cfl}$ respectively. Finally, we combine $AS_{\confl}$ and $AS_{\cfl}$
(presented in subsection~\ref{combining1}) to obtain our final approximate solution $AS_{\concfl}$. For a solution $S$ to an instance $I$, let $Cost_{I}(S)$ denote the cost of $S$, we will drop $I$ wherever it will clear from the context for brevity of notations. Refer Algorithm~\ref{ConCFL_Algo} for details of the algorithm.

\begin{algorithm}[H] 
    \setcounter{AlgoLine}{0}
	\SetAlgoLined
    \DontPrintSemicolon  
    \SetKwInOut{Input}{Input}
	\SetKwInOut{Output}{Output}
    \Input{$I_{\concfl}(\facilityset, \clientset, f, d,c,u)$}
    Create an instance $I_{\confl}(\facilityset, \clientset, f, d, c)$ of \ConFL using $I_{\concfl}$, by dropping the capacities on facilities. Since any solution to instance $I_{\concfl}$ is feasible for $I_{\confl}$ as well, we have $Cost(\Opt{\confl}) \leq Cost(\Opt{\concfl})$.\;

    Create an instance $I_{cfl}(\facilityset, \clientset,f, c, u)$ of \CFL using instance $I_{\concfl}(\facilityset, \clientset, f, d,c,u)$ of \CFL, by dropping the connection cost and the connectivity constraint. Since any solution to $I_{\concfl}$ is feasible for $I_{\cfl}$ as well, we have  $Cost(\Opt{\cfl}) \leq Cost(\Opt{\concfl})$.\;

    Obtain a $\alpha$- approximate solution $AS_{\confl}$ to instance $I_{\confl}$ using approximation algorithm for \ConFL given by Grandoni and Rothvob~\cite{ConFLGrandoni2011}.\;
    
    Obtain a $\beta$- approximate solution $AS_{\cfl}$ to instance $I_{\cfl}$ using approximation algorithm for \CFL given by Bansal~\etal~\cite{Bansal}.\;
    
    Combine solutions $AS_{\confl}$ and $AS_{\cfl}$ to obtain a  solution $AS_{\concfl}= (\facilityset', \phi)$ to $I_{\concfl}$ such that
    {$Cost(AS_{\concfl}) \leq \alpha  Cost(AS_{\confl}) + 2\beta Cost(AS_{\cfl})$}.
	\caption{Algorithm for \ConCFL}
	\label{ConCFL_Algo}
\end{algorithm}

\subsection{Combining $AS_{\confl}$ and $AS_{\cfl}$ }
\label{combining1}
In this section, we combine $AS_{\confl}$ and $AS_{\cfl}$ to obtain an approximate solution $AS_{\concfl}$ for ConCFL. Refer Figure~\ref{fig1}. For every facility $i$ opened in $AS_{\cfl}$, we open $i$ in $AS_{\concfl}$ and assign $j$ to $i$ if it was assigned to $i$ in $AS_{\cfl}$. The cost of opening $i$ and assigning $j$ to $i$ is bounded by cost of $AS_{\cfl}$. The capacities are respected because they are respected in $AS_{\cfl}$.
We connect the facilities opened in $AS_{\cfl}$ to facilities opened in $AS_{\confl}$ which are already connected via a steiner tree. 
We bound this cost (additonal)
as follows: let $i$ be a facility in $\facilityset_{\cfl}$. Let $j$ be a client served by $i$ in $AS_{\cfl}$ and by $i'$ in $AS_{\confl}$ (wlog, we assume that such a client exists for otherwise $i$ can be closed).
The cost of connecting $i$ to $i'$ is bounded by 
$\dist{i}{i'} \leq \dist{i}{j} + \dist{j}{i'}$ (by triangle inequality (refer Figure~\ref{fig2})).
Summing over all $i \in \facilityset_{\cfl}$, the cost is bounded by $S_{\confl} + S_{\cfl}$ where $S_{\confl}$ and $S_{\cfl}$ denote the service cost of $AS_{\confl}$ and $AS_{\cfl}$ respectively.
 

\begin{figure}
   \centering
    \includegraphics[width=8cm]{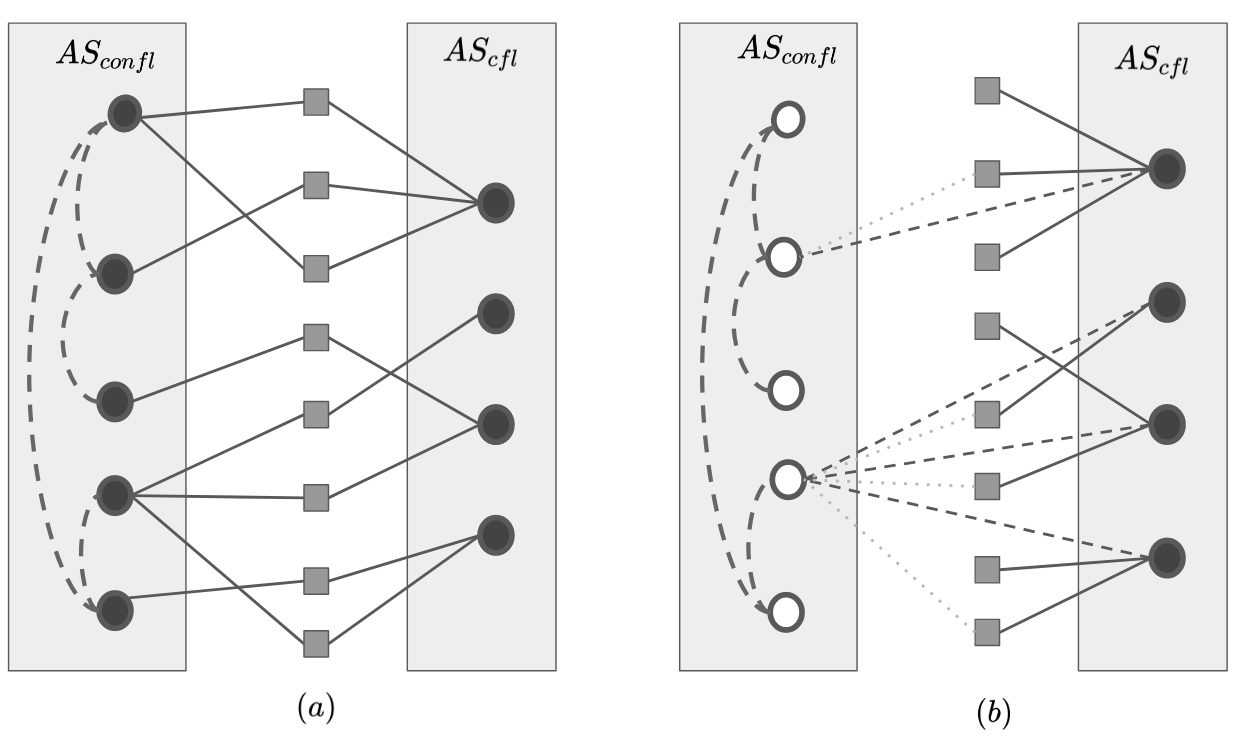}
   \caption{($a$) Represents solution $AS_{\confl}$ and $AS_{\cfl}$ with respective opened facilities(circles) and assignment of clients(squares).  ($b$) Construction of Solution $AS_{\concfl}$. The filled circles represent the opened facilities, the thick lines represent the assignment of clients and the dashed lines represents the connection of opened facilities to the steiner tree.}
    \label{fig1}
\end{figure}


\begin{figure}
    \centering
    \includegraphics[width=5cm]{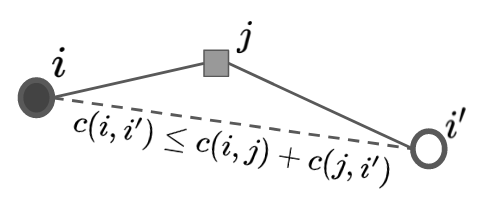}
    \caption{$c(i,i') \leq c(i,j) + c(j,i')$}
    \label{fig2}
\end{figure}

\noindent The overall cost is bounded by $Cost(AS_{\confl}$) + $2 Cost(AS_{\cfl})$ which is $\leq \alpha Cost(\Opt{\confl}) + 2\beta Cost(\Opt{\cfl})$. Using Rothvob and Grandoni~\cite{ConFLGrandoni2011} for $AS_{\confl}$ we have $\alpha=3.19$. 
For non-uniform capacities, using Bansal~\etal~\cite{Bansal} for $AS_{\cfl}$, we have $\beta=5$, which gives us a $13.19$-factor approximation. For uniform capacities, we get $9.19$ factor using $3$-factor approximation of Aggarwal~\etal~\cite{Aggarwal} for $AS_{\cfl}$.

\noindent \textbf{Introducing cardinality:} The same algorithm extends to ConC$k$M/ConC$k$FL/ConC$k$C if we take an instance of connected $k$M/$k$FL/$k$C on one side and capacitated $k$M/$k$FL/$k$C on the other side. Note that the violations in capacities/cardinality are preserved from the underlying solutions. Plugging in the current best results of Con$k$M/Con$k$FL/Con$k$C and C$k$M/C$k$FL/C$k$C, we get the results stated in Table~\ref{tab_results1}. 

\section{Connected, Capacitated and Prize-Collecting varaints of the problem}
\label{sectionConCPFL}

Prize-collecting facility location problem is a generalisation of FL where every client $j$ has an associated penalty cost $p_j$ . We wish to open $\facilityset' \subseteq \facilityset$ of facilities, select $\clipen \subseteq \cliset$ of clients to pay penalty and compute an assignment function $\phi : (\clientset \setminus \clipen) \rightarrow \facilityset'$ for the remaining clients.
Our goal is to minimize the total cost of opening the facilities in $\facilityset'$, paying penalty of clients in $\clipen$ and serving clients in $\clientset \setminus \clipen$. 

In this section, we present a constant factor approximation for \ConCPFL which is a common generalisation of ConPFL and CPFL, that is, we have both the connectivity as well as the capacity constraints in prize-collecting facility location problem. The result is obtained by combining the solutions of \ConPFL and \CPFL (obtained in Section ~\ref{sectionConPFL} and ~\ref{sectionCPFL} respectively). The idea is similar to the one presented in Section~\ref{sectionConCFL}. Let $I_{\concpfl}$ be an instance of ConCPFL. We first make two instances $I_{\conpfl}$ and $I_{\cpfl}$ of \ConPFL and CPFL by dropping the capacity and the connectivity constraints respectively. Note that, an optimal solution $\Opt{\concpfl}$ of \ConCPFL forms feasible solution for instances $I_{\conpfl}$ and $I_{\cpfl}$. Hence the cost of the optimal solutions $\Opt{\conpfl}$ and $\Opt{\cpfl}$ of $I_{\conpfl}$ and $I_{\cpfl}$ respectively are bounded. Next we solve these instance using approximation algorithms for \ConPFL and CPFL to obtain two approximate solutions $AS_{\conpfl}$ and $AS_{\cpfl}$ respectively. Finally, we combine $AS_{\conpfl}$ and $AS_{\cpfl}$ to obtain our final approximate solution $AS_{\concpfl}$. Refer Figure~\ref{fig3}. As in Section~\ref{sectionConCFL} we would like to open all the facilities that were opened in $AS_{\cpfl}$, connect them using the solution to $AS_{\conpfl}$ and serve the clients that were served in $AS_{cpfl}$ paying penalty for the rest. However, we do not know how to bound the cost of connecting the facilities in $AS_{\cpfl}$ to the facilities opened in $AS_{conpfl}$ in this case:  some clients served in $AS_{\cpfl}$ may be paying penalty in $AS_{\conpfl}$. As a result, it is possible that all the clients served by
a facility opened in $AS_{\cpfl}$ 
are paying penalty in $AS_{\conpfl}$; we do not know how to bound the cost of connecting such a facility to the facilities in $AS_{\concpfl}$ (for example, facility $i$ in Figure~\ref{fig3}).
Thus, we decide to pay penalty for clients that were paying penalty in any one of the solutions. The penalty costs of these clients are paid by $AS_{\cpfl}$ and $AS_{\conpfl}$. Next, we look at the remaining clients, open the facilities each of which is serving at least one such client in $AS_{\cpfl}$, assign clients and connect the opened facilities to the facilities in $AS_{\conpfl}$ in the same manner as in Section~\ref{sectionConCFL}. Capacities are respected as they are respected in the underlying problem. The bound on facility cost, service cost and connection cost is obtained in a similar manner as done in section~\ref{sectionConCFL}.

\noindent Using $21.32$-factor approximation for ConPFL and $5$- factor approximation for CPFL (presented in sections~\ref{sectionConPFL} and~\ref{sectionCPFL} resp.), we get a $31.32$-factor approximation for ConCPFL. The result holds for both uniform and non-uniform capacities.

\begin{figure}
    \centering
    \includegraphics[width=8cm]{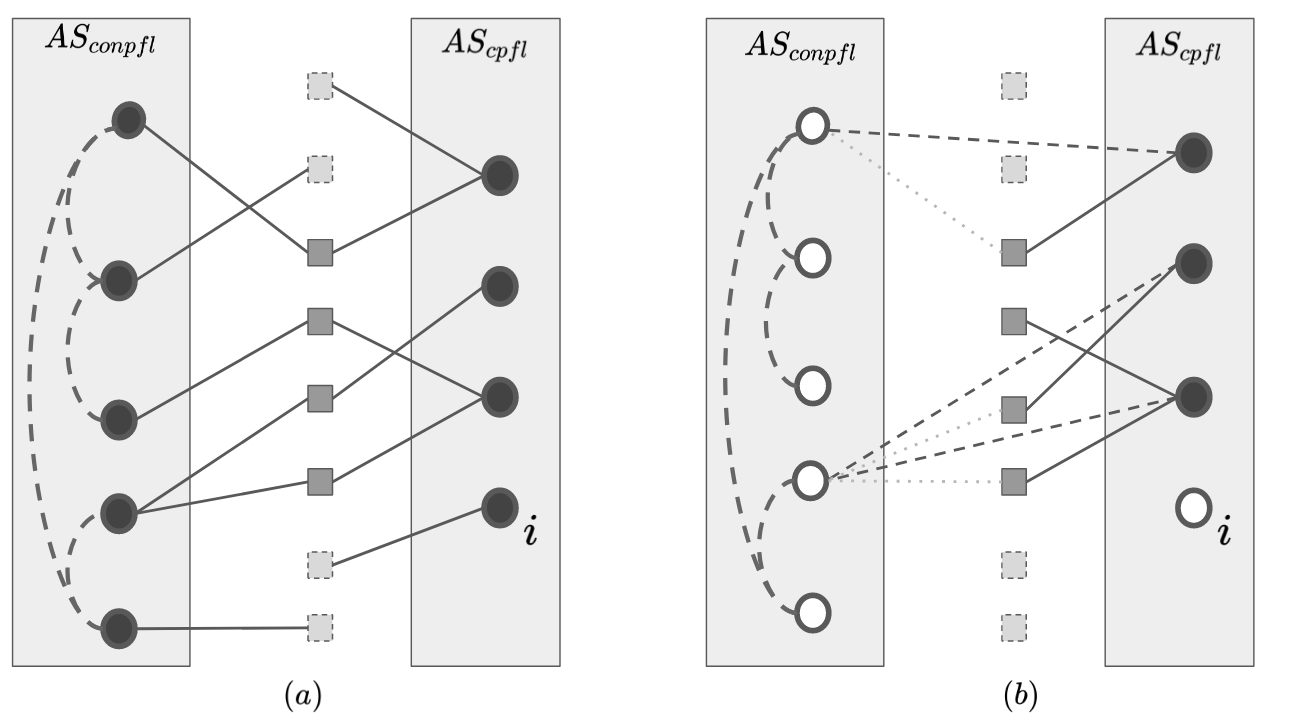}
    \caption{($a$) Represents solution $AS_{\conpfl}$ and $AS_{\cpfl}$ with respective opened facilities(circles) and assignment of clients(squares). The clients that pay penalty in one of the solutions are represented by light grey dashed boundary squares. ($b$) Construction of Solution $AS_{\concpfl}$. The filled circles represent the opened facilities, the thick lines represent the assignment of clients and the dashed lines represents the connection of opened facilities to the steiner tree.}
   \label{fig3}
\end{figure}
 
\noindent \textbf{Introducing cardinality:} The same algorithm extends to ConCP$k$M/ConCP$k$FL if we take an instance of \ConPFL on one side and CP$k$M/CP$k$FL on the other side. The violation in capacities/cardinality is preserved from underlying problem. Using $21.32$-factor approximation (presented in section~\ref{sectionConPFL}) for ConPFL and $O(1/\epsilon)$ with $(2+\epsilon)$ violation in capacities by Dabas~\etal~\cite{DabasCFLPO} for CP$k$M/CP$k$FL, we get a constant factor approximations for (uniform) ConCP$k$M/ConCP$k$FL with $(2+\epsilon)$ factor violation in capacities. Our algorithm will extend to ConCP$k$C if, in future, we have solutions for the underlying problems {\em viz} ConP$k$C and CP$k$C.


\section{Solving ConPFL via reduction to ConFL}
\label{sectionConPFL}
In this section, we present a constant factor approximation for ConPFL using LP rounding techniques and reduction to ConFL. Let us guess a vertex $v$ which is opened as a facility in the optimal solution. We can run our algorithm for all possible choices of $v$ and choose the best among them. For a non-trivial set $S \subseteq \facilityset$, let $\delta(S)$ represents the cut defined by $S$ and $\facilityset \setminus S$. ConPFL can be formulated as the following integer program:

\label{{ConPFL}}
$\text{Minimise} ~Cost(w,x,y,z) = \sum_{i \in \facilityset} w_iy_i + \sum_{i \in \facilityset}\sum_{j \in \cliset} \dist{i}{j}x_{ij} + \sum_{e \in \edgeset} \concost{e} y_e + \sum_{j \in \cliset} p_j z_j $
\begin{eqnarray}
\text{subject to} &  \sum_{i \in \facilityset} x_{ij} + z_j = 1 ~ &\forall j \in \cliset \label{const1}\\ 
&  x_{ij} \leq w_i~& \forall i \in \facilityset,~\forall j \in \cliset \label{const2}\\ 
&  w_v = 1 \label{const3}\\   
&  \sum_{i \in S} x_{ij} \leq \sum_{e \in \delta(S)}y_e~& \forall S \subseteq \facilityset,~v\notin S,~\forall j \in \cliset \label{const4}\\   
& w_i, x_{ij}, y_e, z_j \in \left\lbrace 0,1 \right\rbrace  ~&\forall~i \in \facilityset,~\forall j \in \cliset,~\forall e \in \edgeset \label{const5}
\end{eqnarray}

where variable $w_i$ denotes whether facility $i$ is open or not, $x_{ij}$ indicates if client $j$ is served by facility $i$ or not, $y_e$ denoted whether edge $e$ belongs to steiner tree or not and $z_j$ denotes if client $j$ pay penalty or not. Constraints~\ref{const1} ensure that every client is either served or pays penalty. Constraints~\ref{const2} ensure that a client is assigned to an open facility only. Constraint~\ref{const3} ensures that the guessed facility $v$ is opened. Constraints~\ref{const4} ensure that if a facility is opened then it is connected to $v$. LP-Relaxation of the problem is obtained by allowing the variables $ w_i, x_{ij}, y_e, z_j \in [0, 1]$. Let $\rho^{*} = <w^*,x^*, y^*, z^*>$ denote the optimal solution of the LP and $LP_{opt}$ denote the cost of $\rho^*$.

\subsection{Identifying the Clients that will Pay Penalty ($\cliset'$)}
We first identify a set  $\clipen$ of clients that pay penalty in our solution $AS_{\conpfl}$ via thresholding technique.  
For every client $j$, if $j$ was paying penalty to an extent of at least $1/2$, we add $j$ to $\clipen$. Formally, $\forall j$ such that $z^*_j \geq 1/2$, add $j$ to $\clipen$, set $z'_j=1$ and $x'_{ij}=0$ for every $i$. Note that, this can be done within $2$ factor of penalty cost, that is, $\sum_{j \in \clipen}p_jz'_j \leq 2 \sum_{j \in \clipen}p_j z^*_j$.

\subsection{An instance $I_{\confl}$ of \ConFL :
}
To handle the remaining clients, we create an instance $I_{\confl}$ of ConFL wherein the client set is reduced with $\cliset \setminus \clipen$. An LP for $I_{\confl}$ (LP2) can be formulated as follows:
\label{{ConFL}}
$\text{Minimise} ~Cost(w,x,y) = \sum_{i \in \facilityset} w_iy_i + \sum_{i \in \facilityset}\sum_{j \in \clientset \setminus \clientset'} \dist{i}{j}x_{ij} + \sum_{e \in \edgeset} \concost{e} y_e $
\begin{eqnarray}
\text{subject to} &  \sum_{i \in \facilityset} x_{ij} = 1 ~ &\forall j \in \clientset \setminus \clientset' \label{LP2_const1}\\ 
&  x_{ij} \leq w_i~& \forall i \in \facilityset,~\forall j \in \clientset \setminus \clientset' \label{LP2_const2}\\ 
&  w_v = 1 \label{LP2_const3}\\   
&  \sum_{i \in S} x_{ij} \leq \sum_{e \in \delta(S)}y_e~& \forall S \subseteq \facilityset,~v\notin S~\forall j \in \cliset \setminus \clientset' \label{LP2_const4}\\   
& 0 \leq w_i, x_{ij}, y_e \leq 1   ~&\forall~i \in \facilityset,~\forall j \in \clientset \setminus \clientset',~\forall e \in \edgeset \label{LP2_const5}
\end{eqnarray}


The following lemma shows existence of a feasible fractional solution to $I_{\confl}$ such that the cost is bounded within a constant factor of LP optimal ($LP_{opt}$). 

\begin{lemma}
\label{fs_confl}
There exist a fractional feasible solution $\rho' = <w', x', y'>$ such that cost is bounded by $2(\sum_{i \in \facilityset} f_iw^*_i + \sum_{i \in \facilityset}\sum_{j \in \clientset \setminus \clientset'} \dist{i}{j}x^*_{ij} + \sum_{e \in \edgeset} \concost{e} y^*_e)$.
\end{lemma}

\begin{proof}

Recall that, every client $j \in \cliset \setminus \cliset'$ is served to an extent of at least $1/2$. For all $i \in \facilityset : x^*_{ij}> 0 $, we raise the assignment of $j$ on $i$ proportionately so that $j$ is fully served. Variables $w$ and $y$ are raised accordingly to satisfy constraints~(\ref{LP2_const2}) and constraints~(\ref{LP2_const4}) respectively. Formally, $\forall j \in \cliset \setminus \cliset'$, $\forall i \in \facilityset$ and $\forall e \in \edgeset$, set $(i)$ ${x}'_{ij} = x^*_{ij}/\sum_{i \in \facilityset}x^*_{ij}$, $(ii)$ $w'_i  = \min \{1, w^*_i \cdot \max_{j \in \cliset \setminus \cliset' : x^*_{ij} > 0 } \{{x}'_{ij}/x^*_{ij} \}\}$ and $(iii)$ $y'_e = \max \{ 1, 2y^*_e \}$. 

Next, we will see that $<w', x', y'>$ is a feasible solution to LP of ConFL.
\begin{enumerate}
    \item Constraints~\ref{LP2_const1} are satisfied because for all $j \in \clientset \setminus \clientset'$, $\sum_{i \in \facilityset} x'_{ij} = \sum_{i \in \facilityset} (x^*_{ij}/\sum_{i \in \facilityset}x^*_{ij}) = 1$.
    \item Constraint~\ref{LP2_const2} is satisfied because $w'_v = w^*_v = 1$.
    \item Constraints~\ref{LP2_const3} holds trivially if $w'_i = 1$ as $x'_{ij} \leq 1$. Otherwise, for $x^*_{ij} > 0$, $w'_i \geq (w^*_i/x^*_{ij}) x'_{ij} \geq x'_{ij}$ as $w^*_i/x^*_{ij} \geq 1$
    \item For $y'_e = 1$, constraints~\ref{LP2_const4} holds because $\sum_{i \in S} x'_{ij} \leq 1$. Otherwise, $\sum_{e \in \delta(S)}y'_e = 2\sum_{e \in \delta(S)} y^*_2 \geq 2\sum_{i \in S} x^*_{ij} \geq \frac{\sum_{i \in S} x^*_{ij}}{\sum_{i \in \facilityset} x^*_{ij}} = \sum_{i \in S}x'_{ij}$ where the last inequality follows because $\sum_{i \in \facilityset}x^*_{ij} \geq 1/2$ for all $j \in \clientset \setminus \clientset'$.
\end{enumerate}

\textbf{Cost Bound:} Next we will bound the cost of our feasible solution $\rho'$. Note that, ($i$) $x'_{ij} = x^*_{ij}/\sum_{i \in \facilityset} x^*_{ij} \leq 2 x^*_{ij}$, $(ii)$ $w'_i \leq 2 w^*_i$ and ($iii$) $y'_e \leq 2 y^*_e$. Therefore, $Cost(w',x',y') = \sum_{i \in \facilityset} w'_iy_i + \sum_{i \in \facilityset}\sum_{j \in \clientset \setminus \clientset'} \dist{i}{j}x'_{ij} + \sum_{e \in \edgeset} \concost{e} y'_e  \leq 2 ( \sum_{i \in \facilityset} w^*_iy_i + \sum_{i \in \facilityset}\sum_{j \in \clientset \setminus \clientset'} \dist{i}{j}x^*_{ij} + \sum_{e \in \edgeset} \concost{e} y^*_e)$.
\end{proof}


Next, we solve the instance $I_{\confl}$ using any algorithm that uses LP optimal as a lower bound on the optimal cost. In particular, we use approximation algorithm by Gupta~\etal~\cite{ConFLGupta2001} as black-box to obtain an approximate solution $AS_{\confl}$.  
The solution $AS_{\confl}$ along with clients in $\clipen$ forms an approximate solution $AS_{\conpfl}$ for ConPFL. The final cost bound is as follows: $ Cost(AS_{\conpfl}) = Cost(AS_{\confl}) + \sum_{j \in \clipen} p_j z'_j \leq \beta Cost(\Opt{\confl}) + 2 \sum_{j \in \clipen} p_j z^*_j \leq 2 \beta (\sum_{i \in \facilityset} w^*_iy_i + \sum_{i \in \facilityset}\sum_{j \in \clientset \setminus \clipen} \dist{i}{j}x^*_{ij} + \sum_{e \in \edgeset} \concost{e} y^*_e) + 2 \sum_{j \in \clipen} p_j z^*_j = 2 \beta LP_{opt}$. Using Gupta~\etal~\cite{ConFLGupta2001}, we have $\alpha=10.66$ which gives us $21.32$-factor approximation.
\section{Solving CPFL via reduction to CFL}
\label{sectionCPFL}
In this section, we present a $5$-factor approximation for \CPFL by reducing it to CFL.
%
To create an instance $I_{\cfl}$
of \CFL from an instance $I_{\cpfl}$ of CPFL, we create a set of dummy facilities $\dummyset$ as follows: for every client we create a collocated dummy facility having facility opening cost equal to the penalty cost of the client and $1$ unit of capacity. That is, $\forall j \in \clientset$, add a facility $i_j$ to $\dummyset$ such that $\dist{j}{i_j}=0$, $f_{i_j} = p_j$ and $u_{i_j}=1$. Open a facility $i$ if it is opened in $\Opt{\cfl}$ and assign client $j$ to it if it was assigned to $i$ in $\Opt{\cfl}$ where $\Opt{\cfl}$ is an optimal solution for $I_{\cfl}$. For a client $j$ for which $\Opt{\cfl}$ pays the penalty, we open the facility $i_j$ collocated with $j$ and assign $j$ to it. Clearly the cost of the solution is bounded by the cost of the optimal ($\Opt{cpfl}$).

We obtain an approximate solution $AS_{\cfl}$ to $I_{\cfl}$ using a $\beta$- approximation to CFL.
Next, we create an approximate solution for $I_{\cpfl}$ from $AS_{\cfl}$: open a true facility $i$ if it is opened in $AS_{\cfl}$ and assign client $j$ to it if it was assigned to $i$ in $AS_{\cfl}$; if a dummy facility $i_j$ collocated with a client $j$ is opened in $AS_{\cfl}$ we pay penalty for $j$ in our solution.  Since $u_{i_j} = 1$ wlog we may assume that if $i_j$ is opened in $AS_{\cfl}$ then $j$ is assigned to $i_j$ (for if this is not the case and $j$ is assigned to a facility $i'$ and $j'$ is assigned to $i_j$, we can obtain another solution by assigning $j$ to $i_j$ and $j'$ to $i'$ without increasing the cost). Clearly, the cost of the solution so obtained is bounded by the cost of $AS_{\cfl}$. Using result of Bansal~\etal~\cite{Bansal} to obtain $AS_{\cfl}$, we have a $5$-factor approximation for the problem.

\section{Conclusion and Future Work}
\label{conclusion}
In this paper, we presented constant factor approximations for connected, capacitated and prize-collecting facility location problem and variants. The approximations were obtained using the solutions of the underlying problems as black-box. 

We successfully reduced ConPFL and CPFL to ConFL and CFL respectively. Obtaining similar reductions when we have cardinality constraints is an interesting open problem. 

Though we feel that reducing ConP$k$C and CP$k$C to Con$k$C and C$k$C respectively would be challenging, the techniques of Con$k$C~\cite{liang2016ConkC} and C$k$C~\cite{AnCkC} should be extendable (modifiable) to accommodate penalties. Then our technique in section~\ref{sectionConCPFL} can be used to give an approximation for ConCP$k$C as depicted in Figure 2.

This paper used reductions and combining techniques to introduce notoriously hard capacity constraints to connected- and prize-collecting variants of some important classical problems. It will interesting to see similar results for other constraints, for example, the outlier constraint where you are allowed to leave some specified number of clients unserved. Note that our technique for prize-collecting variants gives similar results for outlier version with $2$-factor violation in outliers; the challenge would be to get rid of this violation.

\bibliography{ref}

\appendix

\end{document}